\documentclass[a4paper,UKenglish,cleveref, autoref, thm-restate]{lipics-v2021}
\pdfoutput=1

\newcommand{\spp}{\text{Sup}}
\newcommand{\sep}{\text{Sep}}
\newcommand{\E}{\mathcal{E}}
\newcommand{\F}{\mathcal{F}}

\title{Lower Bounds on Dynamic Programming for Maximum Weight Independent Set} 

\author{Tuukka Korhonen}{Department of Computer Science, University of Helsinki, Finland \and \url{https://tuukkakorhonen.com}}{tuukka.m.korhonen@helsinki.fi}{https://orcid.org/0000-0003-0861-6515}{}

\authorrunning{T. Korhonen} 

\Copyright{Tuukka Korhonen} 



\category{Track A: Algorithms, Complexity and Games} 

\funding{This work has been financially supported by Academy of Finland (grant 322869).}

\acknowledgements{I wish to thank Prafullkumar Tale for suggesting to generalize the result from planar graphs to $H$-minor-free graphs. I also thank Matti J\"arvisalo, Andreas Niskanen, and anonymous reviewers for helpful comments.}

\nolinenumbers 

\hideLIPIcs  

\EventEditors{Nikhil Bansal, Emanuela Merelli, and James Worrell}
\EventNoEds{3}
\EventLongTitle{48th International Colloquium on Automata, Languages, and Programming (ICALP 2021)}
\EventShortTitle{ICALP 2021}
\EventAcronym{ICALP}
\EventYear{2021}
\EventDate{July 12--16, 2021}
\EventLocation{Glasgow, Scotland (Virtual Conference)}
\EventLogo{}
\SeriesVolume{198}
\ArticleNo{75}

\theoremstyle{definition}
\newtheorem{definition2}[theorem]{Definition}

\begin{document}

\maketitle

\begin{abstract}
We prove lower bounds on pure dynamic programming algorithms for maximum weight independent set (MWIS).
We model such algorithms as tropical circuits, i.e., circuits that compute with $\max$ and $+$ operations.
For a graph $G$, an MWIS-circuit of $G$ is a tropical circuit whose inputs correspond to vertices of $G$ and which computes the weight of a maximum weight independent set of $G$ for any assignment of weights to the inputs.
We show that if $G$ has treewidth $w$ and maximum degree $d$, then any MWIS-circuit of $G$ has $2^{\Omega(w/d)}$ gates and that if $G$ is planar, or more generally $H$-minor-free for any fixed graph $H$, then any MWIS-circuit of $G$ has $2^{\Omega(w)}$ gates.
An MWIS-formula is an MWIS-circuit where each gate has fan-out at most one.
We show that if $G$ has treedepth $t$ and maximum degree $d$, then any MWIS-formula of $G$ has $2^{\Omega(t/d)}$ gates.
It follows that treewidth characterizes optimal MWIS-circuits up to polynomials for \emph{all} bounded degree graphs and $H$-minor-free graphs, and treedepth characterizes optimal MWIS-formulas up to polynomials for all bounded degree graphs.
\end{abstract}

\section{Introduction}
In this paper, we prove lower bounds for \emph{tropical circuits} computing the weight of a maximum weight independent set (MWIS) of a graph.
A tropical circuit is a circuit with Max and Plus operations as gates.
In particular, we consider \emph{MWIS-circuits} of graphs.
An MWIS-circuit of a graph $G$ is a tropical circuit whose inputs correspond to the vertices of $G$ and which computes the weight of a maximum weight independent set of $G$ for any assignment of weights to the inputs.
An MWIS-formula is an MWIS-circuit where each gate has fan-out at most one.

Our motivation for proving lower bounds for MWIS-circuits is that many algorithmic techniques for maximum weight independent set implicitly build an MWIS-circuit of the input graph, and therefore the running time of any algorithm resulting from such a technique is bounded from below by the minimum size of an MWIS-circuit of the graph.
Examples of algorithmic techniques that build MWIS-circuits are dynamic programming over different kinds of decompositions of graphs~\cite{DBLP:journals/dam/ArnborgP89,DBLP:journals/tcs/Bui-XuanTV11,DBLP:conf/esa/FominG20a} and dynamic programming over potential maximal cliques~\cite{DBLP:journals/siamdm/ChudnovskyPPT20,DBLP:journals/siamcomp/FominTV15,DBLP:conf/soda/LokshantovVV14}.
Examples of algorithmic techniques that build MWIS-formulas are branching~\cite{DBLP:conf/focs/GartlandL20,DBLP:journals/dam/Ostergard02} and maximal independent set enumeration~\cite{moon1965cliques}.

\subsection{Our Results}
We prove unconditional lower bounds for sizes of MWIS-circuits and MWIS-formulas parameterized by graph parameters treewidth and treedepth, respectively.
The lower bounds are exponential in treewidth and treedepth, and therefore well-known algorithms yield matching upper bounds for them~\cite{DBLP:journals/dam/ArnborgP89,DBLP:conf/ijcai/FreuderQ85}.
We emphasize that our lower bounds are not worst-case bounds over graph classes, but instead hold for each individual graph.

\paragraph*{MWIS-Circuits and Treewidth}
First we characterize optimal MWIS-circuits of bounded degree graphs.

\begin{theorem}
\label{the:tw1}
Let $G$ be any graph with treewidth $w$ and maximum degree $d$.
Any MWIS-circuit of $G$ has $2^{\Omega(w/d)}$ gates.
\end{theorem}

Theorem~\ref{the:tw1} is optimal up to a factor $d$ in the sense that for each pair $w,d$ we can construct a graph with treewidth $\Omega(w)$ and maximum degree $O(d)$ that admits an MWIS-formula with $d 2^{w/d}$ gates.

Then we extend the result to some graphs that may have high-degree vertices.
A graph $H$ is a \emph{minor} of a graph $G$ if it can be obtained from $G$ by vertex deletions, edge deletions, and edge contractions.
If $H$ can be obtained by only vertex deletions and edge contractions, then it is an \emph{induced minor}.

\begin{theorem}
\label{the:tw2}
Let $G$ be any graph that contains an induced minor with treewidth $w$ and maximum degree $d$.
Any MWIS-circuit of $G$ has $2^{\Omega(w/(d 4^d))}$ gates.
\end{theorem}

In Theorem~\ref{the:tw2} it is essential to require an induced minor instead of a minor because a complete graph with $n$ vertices admits an MWIS-circuit of size $O(n)$, but contains all $n$-vertex graphs as minors.

A graph $G$ is $H$-minor-free if it does not contain the graph $H$ as a minor.

\begin{corollary}
\label{cor:minor}
Let $H$ be any fixed graph and $G$ any $H$-minor-free graph with treewidth $w$.
Any MWIS-circuit of $G$ has $2^{\Omega(w)}$ gates.
\end{corollary}
\begin{proof}
For any fixed $H$, every $H$-minor-free graph of treewidth $w$ contains an $\Omega(w) \times \Omega(w)$-grid as an induced minor~\cite{DBLP:journals/combinatorica/DemaineH08}\footnote{The stated result in~\cite{DBLP:journals/combinatorica/DemaineH08} is that such a grid is a minor, but the same proof works directly to show that the constructed grid minor is also an induced minor. In particular, the proof in~\cite{DBLP:journals/combinatorica/DemaineH08} does not use any edge deletions, and the corresponding result for bounded-genus graphs that it depends on~\cite{DBLP:journals/siamdm/DemaineHT06} is already stated in terms of contraction to a graph that can be turned into a grid by removing vertices without decreasing treewidth by more than a constant factor.}.
An $\Omega(w) \times \Omega(w)$-grid has treewidth $\Omega(w)$ and maximum degree 4, so the result follows from Theorem~\ref{the:tw2}.
\end{proof}

Corollary~\ref{cor:minor} implies a $2^{\Omega(w)}$ lower bound for all planar graphs because planar graphs are $K_5$-minor-free~\cite{kuratowski1930probleme}.

The following corollary follows from Theorem~\ref{the:tw1}, Corollary~\ref{cor:minor}, constant-factor treewidth approximation in $2^{O(w)} n^{O(1)}$ time~\cite{DBLP:journals/jal/RobertsonS86}, and dynamic programming over a tree decomposition~\cite{DBLP:journals/dam/ArnborgP89}.

\begin{corollary}
\label{cor:main}
There is an algorithm which, given a bounded degree or $H$-minor-free graph $G$ whose smallest MWIS-circuit has $\tau$ gates, constructs an MWIS-circuit of $G$ with $\tau^{O(1)}$ gates in $\tau^{O(1)}$ time.
\end{corollary}

In particular, a property analogous to automatizability of proof systems~\cite{DBLP:journals/siamcomp/BonetPR00} holds for MWIS-circuits on bounded degree graphs and $H$-minor-free graphs.

\paragraph*{MWIS-Formulas and Treedepth}
We characterize optimal MWIS-formulas of bounded degree graphs.

\begin{theorem}
\label{the:td}
Let $G$ be any graph with treedepth $t$ and maximum degree $d$.
Any MWIS-formula of $G$ has $2^{\Omega(t/d)}$ gates.
\end{theorem}

Again, Theorem~\ref{the:td} is optimal up to a factor $d$ by the same construction as Theorem~\ref{the:tw1}.
As formulas can be thought of as bounded space analogies of circuits, Theorem~\ref{the:td} gives further evidence (in addition to e.g.~\cite{DBLP:journals/algorithms/ChenRRV18,DBLP:journals/corr/abs-2011-04778,DBLP:journals/toct/PilipczukW18}) supporting that while treewidth is the right parameter for CSP-like problems when equipped with unlimited space, treedepth is the right parameter when dealing with bounded space.

Obtaining a constant-factor single-exponential time parameterized approximation algorithm for treedepth is a well-known open problem~\cite{DBLP:conf/esa/CzerwinskiNP19}, so while we know that the converse of Theorem~\ref{the:td} existentially holds in bounded degree graphs, we currently do not know how to construct such MWIS-formulas without having the treedepth decomposition as an input.

\subsection{Techniques}
Our main circuit complexity tool is an adaptation of a circuit decomposition lemma used in e.g.~\cite{DBLP:journals/siamcomp/Hyafil79,DBLP:journals/jacm/JerrumS82,DBLP:journals/tcs/Valiant80}.
In particular, we show that this lemma can be adapted so that given an MWIS-circuit with $\tau$ gates of a graph with treewidth $w$ it extracts a family of $\tau$ vertex separators each of size $\Omega(w)$.
Once this family has been extracted, the main challenge for proving Theorem~\ref{the:tw1} is to show that if this family of separators is too small, there exists an independent set that intersects all of the separators.
For this we use the lopsided Lovász Local Lemma~\cite{DBLP:journals/dam/ErdosS91}, though we note that more elementary arguments would suffice to prove the theorem with a worse dependency on $d$.
To extend the result from bounded degree graphs to $H$-minor-free graphs we use the minor model of the bounded degree induced minor with high treewidth to further control the structure of these separators.

For MWIS-formulas parameterized by treedepth $t$ we similarly extract a family of $2\tau$ vertex sets each of size $\Omega(t)$ from a $\tau$-gate MWIS-formula, showing that if an independent set intersects all of these vertex sets the formula cannot compute it.
The same application of the Local Lemma is used to prove that such an independent set indeed exists in low degree graphs if $\tau$ is too small.
The argument for extracting the family from the formula is more ad-hoc than the argument for circuits.

\subsection{Related Work}
The convention of modeling dynamic programming algorithms as tropical circuits originates from the recent works of Jukna~\cite{DBLP:journals/mst/Jukna15,DBLP:journals/siamdm/Jukna16}, although some earlier results in monotone arithmetic circuit complexity apply also to tropical circuits~\cite{DBLP:journals/jacm/JerrumS82}.
In general, tropical circuit lower bounds imply lower bounds for monotone arithmetic circuits, but not necessarily the other way around~\cite{DBLP:journals/mst/Jukna15}.
In addition to the works of Jukna, the other works explicitly giving lower bounds for tropical circuits or formulas that we are aware of are~\cite{DBLP:conf/mfcs/MahajanNT17,mahajan2019shortest}.
We are not aware of prior works on lower bounds for tropical circuits or formulas considering maximum weight independent set or the graph parameters treewidth or treedepth.

There are multiple worst-case hardness results related to different formulations of the independent set polynomial.
In~\cite{DBLP:conf/mfcs/BriquelK09} it was shown that the multivariate independent set polynomial is VNP-complete.
The univariate independent set polynomial is \#P-hard to evaluate at every non-zero rational point~\cite{DBLP:conf/stacs/BlaserH08}, and more fine-grainedly its evaluation has $2^{\Omega(n / \log^3 n)}$ worst-case complexity assuming \#ETH~\cite{DBLP:conf/iwpec/Hoffmann10}.

Chv{\'{a}}tal has shown that a certain proof system for maximum independent set which naturally corresponds to branching algorithms requires exponential size proofs on almost all graphs that have the number of edges linear in the number of vertices~\cite{DBLP:journals/siamcomp/Chvatal77}.

Multiple worst-case lower bounds of form $n^{\Omega(w)}$ in limited models of computation for graph homomorphism problems of a pattern graph with treewidth $w$ to a graph with $n$ vertices are known~\cite{DBLP:conf/innovations/AustrinKK19,DBLP:journals/corr/abs-2011-04778,DBLP:journals/siamcomp/LiRR17}.
In particular, recently it was shown that the worst-case monotone arithmetic circuit complexity of homomorphism polynomial is $\Theta(n^{w+1})$, and the worst-case monotone arithmetic formula complexity is $\Theta(n^t)$, where $t$ is the treedepth of the pattern graph~\cite{DBLP:journals/corr/abs-2011-04778}.

Recently, a lower bound of $2^{\Omega(w)}$ was shown for DNNF-compilation of monotone CNFs with primal treewidth $w$ and bounded degree and arity, applying to all such CNFs~\cite{DBLP:journals/mst/AmarilliCMS20}.
We note that after the acceptance of this paper, we became aware of a reduction from MWIS-circuits to DNNFs that allows to prove a weaker version of our Theorem~\ref{the:tw1} via the result of~\cite{DBLP:journals/mst/AmarilliCMS20}.
In particular, the techniques of~\cite{DBLP:journals/mst/AmarilliCMS20} yield an exponent of form $\Omega(w/2^d)$ instead of the best possible $\Omega(w/d)$ given in Theorem~\ref{the:tw1}.

\subsection{Organization}
In Section~\ref{sec:preli} we present preliminaries on graph theory, define MWIS-circuits and prove simple lemmas on them, and discuss the lopsided Lovász Local Lemma and prove a lemma using it.
In Section~\ref{sec:tw} we prove the lower bounds for MWIS-circuits parameterized by treewidth, i.e., Theorems~\ref{the:tw1} and~\ref{the:tw2}.
In Section~\ref{sec:td} we prove the lower bound for MWIS-formulas parameterized by treedepth, i.e., Theorem~\ref{the:td}.
In Section~\ref{sec:opt} we give the construction that shows the optimality of Theorems~\ref{the:tw1} and~\ref{the:td} up to a factor $d$.
We conclude and discuss future work in Section~\ref{sec:disc}.

\section{Preliminaries}
\label{sec:preli}

\subsection{Graphs}
The vertex set of a graph $G$ is denoted by $V(G)$ and the edge set by $E(G)$.
The set of neighbors of a vertex $v$ is denoted by $N(v)$ and the neighborhood of a vertex set $X$ by $N(X) = \bigcup_{v \in X} N(v) \setminus X$.
Closed neighborhoods are denoted by $N[v] = N(v) \cup \{v\}$ and $N[X] = N(X) \cup X$.
The subgraph $G[X]$ induced by a vertex set $X \subseteq V(G)$ has $V(G[X]) = X$ and $E(G[X]) = \{\{u, v\} \in E(G) \mid u \in X \wedge v \in X\}$.
We also use $G \setminus X = G[V(G) \setminus X]$ to denote induced subgraphs.
An independent set of $G$ is a vertex set $I$ such that $G[I]$ has no edges.
In particular, an empty set is an independent set.

A tree decomposition of a graph $G$ is a tree $T$ whose each vertex $i \in V(T)$ corresponds to a bag $B_i \subseteq V(G)$, satisfying that 
\begin{enumerate}
\item $V(G) = \bigcup_{i \in V(T)} B_i$,
\item for each $\{u, v\} \in E(G)$ there is a bag $B_i$ with $\{u, v\} \subseteq B_i$, and
\item for each $v \in V(G)$ the subtree of $T$ induced by bags containing $v$ is connected.
\end{enumerate}
The width of a tree decomposition is $\max |B_i|-1$ and the treewidth $tw(G)$ of a graph $G$ is the minimum width over its tree decompositions.

A treedepth decomposition of a graph $G$ is a rooted forest $F$ with vertex set $V(F) = V(G)$, satisfying for each $\{u, v\} \in E(G)$ that $u$ and $v$ have an ancestor-descendant relation in $F$.
The depth of $F$ is the maximum number of vertices on a simple path from a root to a leaf in $F$.
The treedepth $td(G)$ of a graph $G$ is the minimum depth over its treedepth decompositions.
Note that $tw(G) + 1 \le td(G)$.

\subsection{MWIS-Circuits}
We start by giving a formal definition of a tropical circuit.
Our definition is non-standard in that it does not allow any other input constants than $0$, which we can w.l.o.g. assume in the context of maximum weight independent set.
For a comprehensive treatment of tropical circuits and their relations to monotone Boolean and monotone arithmetic circuits see~\cite{DBLP:journals/mst/Jukna15}. 

\begin{definition2}
A tropical circuit over variables $X$ is a directed acyclic graph with in-degree of each vertex either $0$ or $2$.
The vertices are called gates, the in-degree of a gate is called fan-in, and the out-degree of a gate is called fan-out.
Each gate with fan-in $0$ is labeled with a variable $x_i \in X$ or the constant $0$ and each gate with fan-in $2$ is labeled with either $\max$ or $+$.
One gate is designated as the output gate.
A tropical formula is a tropical circuit where each gate has fan-out at most $1$.
\end{definition2}

With an assignment of real numbers to the variables $X$, a tropical circuit outputs a number computed by the output gate by natural semantics, i.e., a gate labeled with a variable $x_i$ computes the value of $x_i$, a gate labeled with $0$ computes $0$, a gate labeled with $+$ computes the sum of the values computed by its children, and a gate labeled with $\max$ computes the maximum of the values computed by its children.
In particular, a tropical circuit computes a tropical polynomial in the variables $X$ over the tropical $(\mathbb{R} \cup \{-\infty\}, \max, +)$ semiring.
In the tropical semiring $\max$ corresponds to addition and $+$ corresponds to multiplication, with $-\infty$ as the zero and $0$ as the unit.
We will refer to $\max$ as addition and to $+$ as multiplication.

We define an \emph{MWIS-polynomial} with the following simple lemma.

\begin{lemma}
\label{lem:mwispoly}
Let $G$ be a graph.
A tropical circuit over variables $V(G)$ computes the weight of a maximum weight independent set of $G$ for any assignment of real weights to the inputs if and only if for the tropical polynomial $f$ computed by the circuit it holds that
\begin{enumerate}
\item each monomial of $f$ is of form $v_1 \cdot \ldots \cdot v_l$, where $\{v_1, \ldots, v_l\}$ is an independent set of $G$ and
\item for each independent set $\{v_1, \ldots, v_l\}$ of $G$ there is a monomial $v_1 \cdot \ldots \cdot v_l$ in $f$, including the empty independent set corresponding to the empty product $0$.
\end{enumerate}
\end{lemma}
\begin{proof}
For the if-direction, (1) guarantees that the value computed by the circuit is at most the weight of a maximum weight independent set and (2) guarantees that the value is at least the weight of a maximum weight independent set.

For the only if-direction, if some monomial would not be multilinear, i.e., include a factor $v^2$ for some vertex $v$, the output would be incorrect when assigning weight $1$ to $v$ and $0$ to other vertices.
If some monomial would be of form $v_1 \cdot \ldots \cdot v_l$, where $\{v_1, \ldots, v_l\}$ is not an independent set the output would be incorrect when assigning weight $1$ to those $v_i$ and $0$ to others.
Finally, if the output polynomial would not include $v_1 \cdot \ldots \cdot v_l$ as a monomial for some independent set $\{v_1, \ldots, v_l\}$ then the circuit would be incorrect when assigning weight $1$ to vertices of this independent set and $-1$ to others.
\end{proof}

An MWIS-polynomial of a graph $G$ is a polynomial $f$ satisfying (1) and (2) in Lemma~\ref{lem:mwispoly}.
An MWIS-circuit of $G$ is a tropical circuit that computes an MWIS-polynomial of $G$.
An MWIS-formula of $G$ is an MWIS-circuit of $G$ that is a tropical formula.

We note that requiring the circuit to work for all real weights is not a strong assumption: 
Any MWIS-circuit that works for weights $\{0, 1\}$ can be turned into an MWIS-circuit that works for weights $\mathbb{R} \cup \{-\infty\}$ by replacing each input variable $v_i$ by $\max(v_i, 0)$.
In particular, the weight of an empty independent set is $0$, so negative weights will never be used.
Our assumption that the only constant available to the circuit is $0$ can be justified by noting that if an output monomial would contain a positive constant the circuit would be incorrect on the all-zero input, and that if an output monomial would contain a negative constant it should also occur without the constant.
In particular, any other constants than $0$ could be replaced by $0$.

Next we make some simple observations on the structure of MWIS-circuits.

\begin{definition2}
A partial MWIS-polynomial is a polynomial $f$ satisfying (1) in Lemma~\ref{lem:mwispoly}.
A partial MWIS-circuit is a tropical circuit computing a partial MWIS-polynomial.
\end{definition2}

Note that by monotonicity of $(\max, +)$ computations we can assume that each gate of an MWIS-circuit computes a partial MWIS-polynomial and therefore each subcircuit is a partial MWIS-circuit.

\begin{definition2}
Let $f$ be a partial MWIS-polynomial.
We denote by $\spp(f)$ the support of $f$, that is, the variables that occur in the monomials of $f$.
\end{definition2}

We also use $\spp(g)$ for a gate $g$ to denote the support of the polynomial computed by the gate.
Note that each monomial of $f$ corresponds to an independent set of $G[\spp(f)]$.

The following property is the basis for proving lower bounds for MWIS-circuits.

\begin{lemma}
\label{lem:prodpro}
Let $f = g \cdot h$ be a partial MWIS-polynomial of a graph $G$.
The sets $N[\spp(g)]$ and $\spp(h)$ are disjoint.
\end{lemma}
\begin{proof}
If there was a vertex $v \in \spp(g) \cap \spp(h)$ then $f$ would contain a monomial with a factor $v^2$.
If there was a vertex $v \in \spp(g)$ and $u \in \spp(h)$ with $\{u, v\} \in E(G)$, then there would be a monomial in $f$ containing a factor $u \cdot v$.
\end{proof}

We will say that a partial MWIS-polynomial $f$ or a circuit computing $f$ \emph{computes} an independent set $I$ if $f$ contains the monomial $\prod_{v_i \in I} v_i$.
In particular, an MWIS-polynomial computes every independent set.

\subsection{Lopsided Lovász Local Lemma}
The lopsided Lovász Local Lemma~\cite{DBLP:journals/dam/ErdosS91} (see~\cite{DBLP:books/daglib/0021015} for the general version) is a method for showing that there is a non-zero probability that none of the events in a collection of events hold.
In particular, we use it to show that independent sets satisfying certain requirements exist.

\begin{definition2}
Let $\E_1, \ldots, \E_n$ be events in a probability space.
A graph $\Gamma$ is a negative dependency graph of the events if its vertices are $V(\Gamma) = \{\E_1, \ldots, \E_n\}$ and for all events $\E_i$ and subsets $J \subseteq V(\Gamma) \setminus N(\E_i)$ it holds that $\Pr[\bigcup_{j \in J} \E_j \mid \E_i] \ge \Pr[\bigcup_{j \in J} \E_j]$.
\end{definition2}

In words, the negative dependency graph should capture all negative correlations between the events.

\begin{proposition}[\cite{DBLP:books/daglib/0021015}]
\label{pro:llll}
Let $\E_1, \ldots, \E_n$ be a collection of events with a negative dependency graph $\Gamma$.
If there exists real numbers $x_1, \ldots, x_n$ with $0 < x_i < 1$ such that for each $i$ it holds that $\Pr[\E_i] \le x_i \prod_{\E_j \in N(\E_i)} (1-x_j)$, then $\Pr[\bigcap_{i=1}^n \overline{\E_i}] > 0$.
\end{proposition}

\subsection{Hitting Vertex Sets with Independent Sets}
We prove a lemma which captures our use of the Local Lemma in Theorems~\ref{the:tw1} and~\ref{the:td}.
We spell out the constants to emphasize that they are not particularly high, although noting that a more careful proof could improve them a bit.

\begin{lemma}
\label{lem:hitting}
Let $G$ be a graph with maximum degree $d$ and $\F$ a family of vertex subsets of $G$, each member of $\F$ containing at least $k$ vertices.
If $6|\F| \le e^{k/(6d)}$, then there exists an independent set of $G$ that intersects all sets in $\F$.
\end{lemma}
\begin{proof}
We assume $d \ge 2$ as the lemma is easy to verify for $d \le 1$.
We use the Local Lemma to construct such an independent set.
We let each vertex be in the independent set with probability $p = 1/(2d)$.
Our bad events are $\E_e$ for each each edge $e$ indicating that both endpoints of $e$ are selected in the independent set, and $\E_A$ for each $A \in \mathcal{F}$ indicating that the independent set does not intersect $A$.
The negative dependency graph is a bipartite graph connecting $\E_e$ to $\E_A$ if at least one of the endpoints of $e$ is in $A$.
In particular, note that the edge events $\E_e$ have non-negative correlation with each other and the vertex set events $\E_A$ also have non-negative correlation with each other.
For all edge events $\E_e$ we choose $x_e = 1/(3d^2+1)$ and for all vertex set events $\E_A$ we choose $x_A = 1/(5|\F|+1)$.
Now, by Proposition~\ref{pro:llll}, it suffices to verify that 
\begin{equation}
\Pr[\E_e] = p^2 \le x_e (1-x_A)^{|\F|}
\end{equation}
and
\begin{equation}
\Pr[\E_A] = (1-p)^{|A|} \le x_A (1-x_e)^{|A|d}
\end{equation}
hold whenever $6|\F| \le e^{|A|/6d}$.

For (1), a lower bound for the right hand side is $e^{-1/5}/(3d^2+1)$, which can be verified to be greater than $p^2 = 1/(4d^2)$ when $d\ge2$.
For (2), an upper bound for the left hand side is $e^{-|A|/(2d)}$, and a lower bound for the right hand side is $x_A e^{-|A|d/(3d^2)}$, implying that (2) holds if $e^{-|A|/(2d)} e^{|A|/(3d)} \le x_A$.
This simplifies to $e^{-|A|/(6d)} \le x_A \Leftrightarrow e^{|A|/(6d)} \ge 5|\F|+1$.
\end{proof}

\section{Treewidth and MWIS-Circuits}
\label{sec:tw}
In this section we prove lower bounds for MWIS-circuits parameterized by treewidth, i.e., Theorems~\ref{the:tw1} and~\ref{the:tw2}.

We use a witness of high treewidth due to Robertson-Seymour treewidth approximation algorithm~\cite{DBLP:journals/jal/RobertsonS86}.
A separation of a graph $G$ is an ordered triple of vertex sets $(A,S,B)$ such that $A,S,B$ are disjoint, $A \cup S \cup B = V(G)$, and no vertex of $A$ is adjacent to a vertex of $B$.
The order of a separation $(A,S,B)$ is $|S|$.
A separation $(A,S,B)$ is a balanced separation of a vertex set $X \subseteq V(G)$ if $|A \cap X| \le 2|X|/3$ and $|B \cap X| \le 2|X|/3$.

\begin{lemma}[\cite{DBLP:journals/jal/RobertsonS86}]
\label{lem:twlb}
If a graph $G$ has treewidth at least $4k$, then there is a vertex set $X \subseteq V(G)$ such that any balanced separation of $X$ in $G$ has order at least $k$.
\end{lemma}

The next lemma is our main tool to connect circuit complexity with treewidth.
This lemma is an adaptation of a classical circuit decomposition lemma (e.g. Theorem~1 in~\cite{DBLP:journals/siamcomp/Hyafil79}, Lemma~3 in~\cite{DBLP:journals/tcs/Valiant80}).
In our applications the vertex set $X$ will be the set given by Lemma~\ref{lem:twlb}.

\begin{lemma}
\label{lem:decomp}
Let $G$ be a graph and $X \subseteq V(G)$ with $|X| \ge 2$.
If there is an MWIS-circuit of $G$ with $\tau$ gates, then we can write an MWIS-polynomial of $G$ as $g_1 \cdot h_1 + \ldots + g_\tau \cdot h_\tau$, where for all $i$ it holds that $|\spp(g_i) \cap X| \le 2|X|/3$ and $|\spp(h_i) \cap X| \le 2|X|/3$.
\end{lemma}
\begin{proof}
Let $f + e$ be an MWIS-polynomial of $G$, where $f$ can be computed by a tropical circuit with $\tau$ gates. (The term $e$ is here for the induction argument. In the first step we can assume it to be empty.)
We will show that there is an MWIS-polynomial $f' + g \cdot h + e$ of $G$, where $f'$ can be computed by a tropical circuit with $\tau-1$ gates, and $|\spp(g) \cap X| \le 2|X|/3$ and $|\spp(h) \cap X| \le 2|X|/3$.
The lemma follows from this by induction.

If $|\spp(f) \cap X| \le 2|X|/3$ we are done.
Otherwise, we traverse the circuit computing $f$ down starting from the output gate, always choosing the one of the two child gates whose support has larger intersection with $X$, until we reach a gate $v$ computing a polynomial $f_v$ with $|X|/3 \le |\spp(f_v) \cap X| \le 2|X|/3$.
Let $f_{v=-\infty}$ be the polynomial computed by the circuit when the value of the gate $v$ is set to $-\infty$.
Now we can write an MWIS-polynomial of $G$ as $f_{v=-\infty} + f_v \cdot g + e$, for example by letting $g$ be an MWIS-polynomial of $G \setminus N[\spp(f_v)]$.
Now, we observe that $f_{v=-\infty}$ can be computed by a circuit with $\tau-1$ gates.
We also observe that the supports of $f_v$ and $g$ cannot intersect, and therefore $|\spp(g) \cap X| \le 2|X|/3$.
\end{proof}

\subsection{Proof of Theorem~\ref{the:tw1}}
Now we complete the proof of Theorem~\ref{the:tw1} by putting Lemmas~\ref{lem:hitting}, \ref{lem:twlb}, and~\ref{lem:decomp} together.

\begin{lemma}
Let $G$ be a graph with maximum degree $d$ and treewidth at least $4k$.
Any MWIS-circuit of $G$ has at least $e^{k/(6d)}/6$ gates.
\end{lemma}
\begin{proof}
Suppose there is an MWIS-circuit of $G$ with $\tau$ gates.
By Lemma~\ref{lem:twlb} there is a vertex set $X \subseteq V(G)$ that does not admit a balanced separation of order less than $k$.
By Lemma~\ref{lem:decomp} we can write an MWIS-polynomial of $G$ as $g_1 \cdot h_1 + \ldots + g_\tau \cdot h_\tau$, where for all $i$ it holds that $|\spp(g_i) \cap X| \le 2|X|/3$ and $|\spp(h_i) \cap X| \le 2|X|/3$.
Now, by Lemma~\ref{lem:prodpro} each multiplication $g_i \cdot h_i$ defines a balanced separation $(\spp(g_i), V(G) \setminus \spp(g_i \cdot h_i), \spp(h_i))$ of $X$.
The order of such a separation is $|V(G) \setminus \spp(g_i \cdot h_i)|$, and therefore $|V(G) \setminus \spp(g_i \cdot h_i)| \ge k$.
Note that $g_i \cdot h_i$ does not compute an independent set $I$ if $I$ intersects $V(G) \setminus \spp(g_i \cdot h_i)$.
Therefore, by letting $\F$ be the collection of vertex sets $\{V(G) \setminus \spp(g_1 \cdot h_1), \ldots, V(G) \setminus \spp(g_\tau \cdot h_\tau)\}$, Lemma~\ref{lem:hitting} shows that if $6 \tau \le e^{k/(6d)}$ we can construct an independent set that is not computed by any of the multiplications, contradicting the assumption that we have an MWIS-circuit.
\end{proof}

\subsection{Proof of Theorem~\ref{the:tw2}}
An \emph{induced minor model} of a graph $H$ in a graph $G$ is a function $f : V(H) \to 2^{V(G)} \setminus \{\emptyset\}$, where $2^{V(G)}$ denotes the power set of $V(G)$, satisfying that
\begin{enumerate}
\item the sets $f(u)$ and $f(v)$ are disjoint for $u \neq v$,
\item for each $v \in V(H)$ the induced subgraph $G[f(v)]$ is connected, and
\item $\{u, v\} \in E(H)$ if and only if $N(f(u))$ intersects $f(v)$.
\end{enumerate}

A graph $G$ contains a graph $H$ as an induced minor if and only if there is an induced minor model of $H$ in $G$.
For $v \in V(H)$ we call the induced subgraphs $G[f(v)]$ \emph{clusters}.

First, we ensure that the maximum degree of each cluster is bounded.
\begin{lemma}
\label{lem:lowdg_min}
Let $G$ be a graph that contains a graph $H$ with maximum degree $d$ as an induced minor.
There is an induced minor model $f$ of $H$ in $G$ such that the maximum degree of each cluster $G[f(v)]$ is at most $d$.
\end{lemma}
\begin{proof}
Consider an induced minor model $f$ of $H$ in $G$ and a cluster $G[f(v)]$ for some $v \in V(H)$.
Because the degree of $H$ is at most $d$, we can assign the cluster a set of at most $d$ terminal vertices whose connectivity should be preserved in order to satisfy that $f$ is an induced minor model of $H$ in $G$.
Now, we can remove from the cluster any vertices as long as the terminals stay connected.
In particular, if there is a vertex $u$ with degree $>d$ in $G[f(v)]$, then we can consider the shortest paths from $u$ to the terminals, and remove from $G[f(v)]$ the vertices of $N(u) \cap G[f(v)]$ that do not participate in the shortest paths.
This makes the degree of $u$ in $G[f(v)]$ at most $d$.
\end{proof}

We also need the following lemma.
\begin{lemma}
\label{lem:isbasic}
Let $I$ be an independent set selected uniformly at random from the set of all independent sets of a graph $G$ with maximum degree $d$.
For all $v \in V(G)$ it holds that $\Pr[v \in I] \ge 1/2^{d+1}$.
\end{lemma}
\begin{proof}
For any set $J \subseteq N(v)$ it holds that $\Pr[I \cap N(v) = J] \le \Pr[I \cap N(v) = \emptyset]$ because we can map any independent set $I$ with $I \cap N(v) = J$ into an independent set $I \setminus N(v)$.
Therefore $\Pr[I \cap N(v) = \emptyset] \ge 1/2^d$, so by observing that $\Pr[v \in I \mid I \cap N(v) = \emptyset] \ge 1/2$ we get $\Pr[v \in I] \ge 1/2^{d+1}$.
\end{proof}

Next we finish the proof with similar arguments as in the proof of Theorem~\ref{the:tw1}, but with a different kind of construction of the independent set with the Local Lemma.
In this case the constants involved appear to be impractical.

\begin{lemma}
\label{lem:con_res}
Let $G$ be a graph that contains a graph $H$ with maximum degree $d$ and treewidth $4k$ as an induced minor. Any MWIS-circuit of $G$ has $2^{\Omega(k/(d4^d))}$ gates.
\end{lemma}
\begin{proof}
Let $f$ be the induced minor model of $H$ in $G$.
First, by Lemma~\ref{lem:lowdg_min} we can assume that the maximum degree of each cluster $G[f(v)]$ is at most $d$.
Now, by Lemma~\ref{lem:twlb} we let $X'$ be a vertex set of $H$ that has no balanced separation of order less than $k$.
Then we let $X$ be a vertex set of $G$ created by mapping each $v \in X'$ to an element of $f(v)$.
For each balanced separation $(A,S,B)$ of $X$ in $G$, the set $S$ must intersect at least $k$ different clusters, because otherwise we could map it into a balanced separation of $X'$ of order $< k$ in $H$.
Therefore, by assuming that $G$ has an MWIS-circuit with $\tau$ gates and applying Lemma~\ref{lem:decomp} with the set $X$ we write an MWIS-polynomial of $G$ as $g_1 \cdot h_1 + \ldots + g_\tau \cdot h_\tau$, observing that for each $i$ the set $S_i = V(G) \setminus \spp(g_i \cdot h_i)$ intersects at least $k$ different clusters.
Now it remains to show that if $\tau$ is too small we can construct an independent set of $G$ that intersects $S_i$ for all $i$.

By removing vertices from each $S_i$ we can assume that $S_i$ contains only vertices in clusters, and moreover contains exactly one vertex from each cluster that it intersects.
We use the Local Lemma to construct the independent set.
First we select each cluster independently with probability $p = 1/(4d 2^d)$, and then for each selected cluster $G[f(v)]$ we select an independent set uniformly at random from the set of all independent sets of $G[f(v)]$.
By Lemma~\ref{lem:isbasic} each vertex of $G$ that is in some cluster will appear in the independent set with probability at least $p/2^{d+1}$. 
Vertices in different clusters appear in it independently of each other.

Now our bad events are $\E_{\{u, v\}}$ for all $\{u, v\} \in E(H)$ indicating that both clusters $G[f(u)]$ and $G[f(v)]$ have been selected and $\E_i$ for each $S_i$ indicating that the set $S_i$ does not intersect the independent set.
Our negative dependency graph has edges connecting each $\E_{\{u,v\}}$ to each $\E_i$ such that $S_i$ intersects $f(u)$ or $f(v)$.
It also has all edges between all events $\E_i$ because $\E_i$ and $\E_j$ can be negatively correlated if $S_i$ and $S_j$ intersect a common cluster.

For edges $\{u, v\} \in E(H)$ we let $x_{\{u, v\}} = 1/(15d^2 4^d + 1)$ and for sets $S_i$ we choose $x_i = 1/(20\tau + 1)$.
Now it suffices to verify that
\begin{equation}
\Pr[\E_{\{u, v\}}] = p^2 \le x_{\{u, v\}} (1-x_i)^\tau
\end{equation}
 and
\begin{equation}
\Pr[\E_i] \le (1-p/2^{d+1})^{|S_i|} \le x_i (1-x_i)^\tau (1-x_{\{u, v\}})^{|S_i|d}
\end{equation}
hold whenever $30 \tau \le e^{7|S_i|/(120d 4^d)}$.
We also assume that $d \ge 3$ since if $d \le 2$ then the treewidth of $H$ is at most $2$.

For (3), a lower bound for the right hand side is $e^{-1/20} / (15d^2 4^d + 1)$, which is greater than $p^2 = 1/(16d^2 4^d)$ when $d \ge 2$.
For (4), a lower bound for the right hand side is $x_i e^{-1/20} e^{-|S_i|d/(15d^2 4^d)}$ and an upper bound for the left hand side is $e^{-|S_i|/(8d 4^d)}$, so it holds whenever
$e^{-|S_i|/(8d 4^d)} \le x_i e^{-1/20} e^{-|S_i|d/(15d^2 4^d)}$ holds, which we can simplify to
$e^{|S_i|(1/(15d 4^d)-1/(8d 4^d))} \le x_i e^{-1/20}$,
and finally to
$e^{-7|S_i|/(120d 4^d)} \le 1/(20\tau + 1) e^{-1/20}$, which holds whenever $30 \tau \le e^{7|S_i|/(120d 4^d)}$.
\end{proof}

\section{Treedepth and MWIS-Formulas}
\label{sec:td}
For treedepth we are not aware of linear high-treedepth witnesses similar to what Lemma~\ref{lem:twlb} is for treewidth.
However, it turns out that we can use very basic properties of treedepth decompositions to establish the connection to formula complexity.

Recall that we denote the treedepth of a graph $G$ with $td(G)$.
The following properties follow from the definition of treedepth.

\begin{proposition}
\label{pro:tdbasic}
Let $G$ be a graph with treedepth $td(G)$.
It holds that
\begin{enumerate}
\item $td(G \setminus \{v\}) \ge td(G)-1$ for any $v \in V(G)$ and
\item $td(G)$ is the maximum of $td(G[C])$ over the connected components $C$ of $G$.
\end{enumerate}
\end{proposition}

For our proof we need to introduce two definitions on MWIS-formulas.
We start by defining \emph{typical} independent sets of a partial MWIS-formulas.

\begin{definition2}
Let $F$ be a partial MWIS-formula of a graph $G$.
An independent set $I$ of $G$ is a typical independent set of $F$ if for each multiplication gate $g$ with $td(G[\spp(g)]) \ge td(G)/2$ it holds that $I$ intersects a connected component $C$ of $G[\spp(g)]$ with $td(G[C]) = td(G[\spp(g)])$.
\end{definition2}

Note that by the property~2 of Proposition~\ref{pro:tdbasic} such component indeed exists.

We also define the \emph{separator} $\sep(g)$ of a gate $g$.
Note that an MWIS-formula forms a tree rooted at the output gate, so we will use standard tree terminology (parent, child, ancestor, descendant).

\begin{definition2}
The separator of the output gate $o$ is $\sep(o) = V(G) \setminus \spp(o)$.
The separator of a gate $g$ whose parent $p$ is a multiplication gate is $\sep(g) = \sep(p)$.
The separator of a gate $g$ whose parent $p$ is a sum gate is $\sep(g) = \sep(p) \cup \spp(p) \setminus \spp(g)$.
\end{definition2}

With the definitions of typical independent sets and separators of gates, we can state the following lemma which will be applied with Lemma~\ref{lem:hitting} to prove our lower bound.

\begin{lemma}
Let $G$ be a graph with $td(G) \ge 2$ and $F$ a partial MWIS-formula of $G$.
If $I$ is a typical independent set of $F$ and intersects $\sep(g)$  for each gate $g$ with $|\sep(g)| \ge td(G)/2$, then $F$ does not compute $I$.
\end{lemma}
\begin{proof}
Let $F$ be such a formula and $I$ such an independent set.
We say that a gate $g$ of $F$ is redundant if $F$ computes $I$ if and only if $F$ without $g$ computes $I$.
First, note that all gates $g$ such that $I$ intersects $\sep(g)$ are redundant because by the definition of separator there is an ancestor gate $g'$ of $g$ with a sum gate parent $p$ such that none of the monomials $M$ contributed from $g'$ to $p$ have $M = \prod_{v_i \in I \cap \spp(p)} v_i$, implying that $g'$ is redundant and thus all of its descendants are redundant.

Now, we prove by induction starting from the leaves that every gate $g$ of $F$ for which $|\sep(g)| + td(G[\spp(g)]) \ge td(G)$ holds is redundant.
First, for all such gates $g$ with $td(G[\sup(g)]) \le td(G)/2$, including all leaves, we have that $|\sep(g)| \ge td(G)/2$, making $g$ redundant by our definition of $I$.
For a sum gate $g$ and its child $c$ we have by property~1 of Proposition~\ref{pro:tdbasic} that $td(G[\spp(c)]) \ge td(G[\spp(g)])-|\spp(g) \setminus \spp(c)|$, rearranging to $td(G[\spp(c)] \ge td(G[\spp(g)]) - |\sep(c)| + |\sep(g)|$, and finally to $td(G[\spp(c)]) + |\sep(c)| \ge td(G[\spp(g)]) + |\sep(g)|$. This implies that if $|\sep(g)| + td(G[\spp(g)]) \ge td(G)$ then both children of $g$ are redundant, making $g$ redundant.
For a multiplication gate $g$ with $td(G[\spp(g)]) \ge td(G)/2$ it follows from the typicality assumption that there is a child $c$ of $g$ with $td(G[\spp(c)]) + |\sep(c)| = td(G[\spp(g)]) + |\sep(g)|$ such that $g$ is redundant if $c$ is redundant.
Therefore the induction works, and because $|\sep(o)| + td(G[\spp(o)]) \ge td(G)$ holds for the output gate $o$, the output gate is redundant and therefore the formula does not compute $I$.
\end{proof}

Now the only thing left to complete the proof of Theorem~\ref{the:td} is to show that if a formula has less than $2^{\Omega(td(G)/d)}$ gates then we can construct an independent set that is typical for the formula and intersects $\sep(g)$ whenever $|\sep(g)| \ge td(G)/2$.
For an independent set to be typical it suffices that it intersects $\spp(g)$ for all gates $g$ with $|\spp(g)| \ge td(G)/2$.
Therefore it suffices to apply Lemma~\ref{lem:hitting} with $\F$ consisting of $\sep(g)$ for all $|\sep(g)| \ge td(G)/2$ and $\spp(g)$ for all $|\spp(g)| \ge td(G)/2$. 
This yields a lower bound of $e^{td(G)/(12d)}/12$ for the number of gates.

\section{Optimality of Theorems~\ref{the:tw1} and~\ref{the:td}}
\label{sec:opt}
We show that for each pair $w,d$ we can construct a graph with treewidth $\Omega(w)$ and maximum degree $O(d)$ that admits an MWIS-formula with $d 2^{w/d}$ gates.

If $d > w$ then a $d$-clique does the job.
Otherwise, we take a bounded degree expander with $w/d$ vertices, having treewidth $\Omega(w/d)$, constructible by e.g.~\cite{DBLP:journals/jcss/GabberG81}.
We replace each vertex of the expander with a $d$-clique (which will be referred to as cluster) such that each vertex of a cluster is connected to each vertex of the clusters of the adjacent vertices.
We denote the constructed graph with $G_{w,d}$

\begin{proposition}
The graph $G_{w,d}$ has treewidth $\Omega(w)$, maximum degree $O(d)$, and admits an MWIS-formula with $d 2^{w/d}$ gates.
\end{proposition}
\begin{proof}
The maximum degree is at most $(d+1)$ times the maximum degree of the original bounded degree expander.
The treewidth is $\Omega(w)$ because if a balanced separator contains one vertex from a cluster it must contain all vertices of the cluster.

Note that by a simple recursion any $n$-vertex graph admits an MWIS-formula with at most $2^n$ gates, so the original expander admits an MWIS-formula with $2^{w/d}$ gates.
We can construct an MWIS-formula of $G_{w,d}$ by taking the MWIS-formula of the original expander and replacing each leaf corresponding to a vertex $v$ with a $d$-gate construction computing the maximum over the vertices of the cluster of $v$.
\end{proof}

\section{Conclusions and Future Work}
\label{sec:disc}
We investigated the tropical circuit complexity of maximum weight independent set.
Our initial motivation for this was the fact that lower bounds for tropical circuits imply lower bounds for many actual algorithmic techniques for maximum weight independent set that are widely used in both theory and practice.
We showed that in bounded degree graphs optimal MWIS-circuits are characterized by treewidth and optimal MWIS-formulas are characterized by treedepth.
We generalized the result for MWIS-circuits to apply beyond bounded degree graphs, to a graph class that includes all planar graphs, and more generally all $H$-minor-free graphs.
The constants hidden by the $\Omega$-notation in Theorems~\ref{the:tw1} and~\ref{the:td} are somewhat practical even though we did not specifically optimize them.
For example, Theorem~\ref{the:tw1} shows that any MWIS-circuit of the $5000 \times 5000$-grid has at least $10^{21}$ gates.

We identify some technical barriers for extending the results.
First, we note that Lemma~\ref{lem:hitting} is not effective in graphs with maximum degree higher than $k$: If $|N(v)| \ge k$, we can add $N(v)$ to $\F$ to force the independent set to avoid $v$, essentially forcing us to work with $G \setminus \{v\}$.
Indeed an example of a graph with high treewidth and no small MWIS-circuits for which Lemma~\ref{lem:hitting} is unsuitable is a clique with each edge subdivided.
In some cases, including $H$-minor-free graphs and the subdivided clique, this barrier can be circumvented with Theorem~\ref{the:tw2} by using a bounded degree induced minor with high treewidth.
We also note that our proofs do not exploit the fact that the separators given by Lemma~\ref{lem:decomp} are balanced beyond just the size bound.

The subdivided clique does not exclude any fixed graph as a minor, so the fact that Theorem~\ref{the:tw2} works also for proving a lower bound for it seems to indicate that Theorem~\ref{the:tw2} is more powerful than what is captured by Corollary~\ref{cor:minor}.
We are in fact not aware of graph families for which a $2^{\Omega(w)}$ lower bound can be proved but Theorem~\ref{the:tw2} does not apply.

An interesting general direction for future work could be to prove Corollary~\ref{cor:main} for as large graph classes as possible, starting by extending the $2^{\Omega(w)}$ lower bound as far as possible.
In particular, $H$-topological-minor-free graphs generalize both bounded degree and $H$-minor-free graphs~\cite{DBLP:journals/ejc/NesetrilM11a}, so proving a $2^{\Omega(w)}$ lower bound for them seems like a natural next step.
Even more generally, it could be that such a lower bound could even apply to all bounded degeneracy graphs.
We hope that this line of work will lead to new insights on the structure of independent sets that could even be useful for positive results on algorithms for maximum weight independent set.



\bibliography{paper}

\end{document}